\documentclass[journal,onecolumn,12pt]{IEEEtran}
\usepackage{pslatex}
\usepackage{amsfonts,color,morefloats}
\usepackage{amssymb,amsmath,latexsym,amsthm}

\newcommand{\lcm}{{\rm lcm}}

\newcommand{\ord}{{\mathrm{ord}}}
\newcommand{\Z}{\mathbb{{Z}}}

\newcommand{\gf}{{\mathrm{GF}}}

\newcommand{\h}{{\text{H}}}

\newtheorem{theorem}{Theorem}
\newtheorem{lemma}[theorem]{Lemma}
\newtheorem{proposition}[theorem]{Proposition}

\newtheorem{definition}{Definition}
\newtheorem{example}{Example}

\begin{document}

\title{On Hermitian LCD codes from cyclic codes and their applications to orthogonal direct sum masking}

\author{Chengju~Li~% <-this % stops a space

\thanks{C. Li is with the Shanghai Key Laboratory of Trustworthy Computing, East China Normal University,
Shanghai, 200062, China (email: cjli@sei.ecnu.edu.cn). The work was supported by the NSFC-Zhejiang Joint Fund for the Integration of Industrialization and Informatization under Grant No. U1509219.

}% <-this % stops a space
% <-this % stops a space
}

\date{\today}
\maketitle

\begin{abstract}
Cyclic codes are an interesting type of linear codes and have wide
applications in communication and storage systems due to their efficient encoding and
decoding algorithms. It was proved that asymptotically good Hermitian LCD codes exist. The objective of this paper is to construct some cyclic Hermitian LCD  codes over finite fields and analyse their
 parameters. The dimensions of these codes are settled and the lower bounds on their minimum distances are presented. Most Hermitian LCD codes presented in this paper are not BCH codes.
 In addition, we employ Hermitian LCD codes to propose a Hermitian orthogonal direct sum masking scheme that achieves protection against fault injection attacks. It is shown that the codes with
 great minimum distances are desired to improve the resistance.

\end{abstract}


\section{Introduction}
Throughout this paper, let $\gf(Q)$ be a finite field of size $Q=q^2$, where $q$ is a prime power. Let
$n$ be an integer with $\gcd(n, q)=1$. An $[n,k, d]$ linear code $\mathcal C$ over $\gf(Q)$ is a linear subspace of $\gf(Q)^n$ with dimension $k$ and
minimum (Hamming) distance $d$.
For any $x \in \gf(Q)$, the conjugate of $x$ is defined by $\bar{x}=x^q$ .

\begin{definition} A linear code $\mathcal C$ over $\gf(Q)$ is called a Hermitian LCD
code (linear code with Hermitian complementary dual) if $\mathcal C \oplus \mathcal C^{\perp_\h}=\gf(Q)^n$(i.e.,
$\mathcal C \cap \mathcal C^{\perp_\h}=\{\textbf{0}\}$), where $\mathcal C^{\perp_\h}$ denotes the Hermitian dual of $\mathcal C$ and is defined by
$$\mathcal C^{\perp_\h}=\{(b_0, b_1, \ldots, b_{n-1}) \in \gf(Q)^n: \sum_{i=0}^{n-1}b_i\bar{c_i}=0 \mbox{ for all } (c_0, c_1, \ldots, c_{n-1}) \in \mathcal C\}.$$

Moreover, a linear code $\mathcal C$ over $\gf(Q)$ is called an \emph{Euclidean LCD
code} if $\mathcal C \oplus \mathcal C^{\perp}=\gf(Q)^n$, where $\mathcal C^{\perp}$ is the Euclidean dual of $\mathcal C$.
\end{definition}

A linear code $\mathcal C$ is called \emph{cyclic} if
$(c_0, c_1, \ldots, c_{n-1}) \in \mathcal C$ implies
$(c_{n-1}, c_0, c_1, \ldots, c_{n-2}) \in \mathcal C$. By identifying any vector $(c_0, c_1, \ldots, c_{n-1}) \in \gf(Q)^n$ with
$$c_0+c_1x+c_2x^2+\cdots+c_{n-1}x^{n-1} \in \gf(Q)[x]/(x^n-1),$$
a code $\mathcal C$ of length $n$ over $\gf(Q)$ corresponds to a subset of
$\gf(Q)[x]/(x^n-1)$.
Then $\mathcal C$ is a cyclic code if and only if the
corresponding subset is an ideal of $\gf(Q)[x]/(x^n-1)$. Note that every
ideal of $\gf(Q)[x]/(x^n-1)$ is principal. Then there is a monic polynomial $g(x)$ of the smallest
degree such that $\mathcal C=\langle g(x) \rangle$ and $g(x) \mid (x^n-1)$. In addition, $g(x)$ is unique and called the \emph{generator
polynomial}, and $h(x)=(x^n-1)/g(x)$ is referred to as the \emph{check polynomial} of $\mathcal C$.
Denote $m=\ord_n(Q)$. Let $\alpha$ be a generator of $\gf(Q^m)^*$ and put $\beta=\alpha^{\frac {Q^m-1} n}$. Then
$\beta$ is a primitive $n$-th root of unity. The set $S=\{0 \le i \le n-1 : g(\beta^i)=0\}$ is referred to as the
\emph{defining set} of $\mathcal C$.

\begin{lemma} \label{LBCH} (BCH bound) Let $\mathcal C$ be a cyclic code of length $n$ over $\gf(Q)$ with defining set
$S$. Suppose that $S$ contains $\delta-1$ consecutive elements for
some positive integer $\delta$. Then $d \ge \delta$, where $d$ is the minimum distance of $\mathcal C$.
\end{lemma}

Let $m_i(x)$ denote the minimal polynomial of $\beta^i$ over $\gf(Q)$. We use $i \bmod n$ to denote
the unique integer in the set $\{0, 1, \ldots, n-1\}$, which is congruent to $i$ modulo $n$.
 We also assume that $m_i(x):=m_{i \bmod n}(x)$. We begin to recall the definition of BCH code. For any integer $\delta$ with $2 \leq \delta \leq n$, define
\begin{equation} \label{GPoly} g_{(n,\delta,b)}(x)=\lcm(m_{b}(x), m_{b+1}(x), \cdots, m_{b+\delta-2}(x)),\end{equation}
where $\lcm$ denotes the least common multiple of these polynomials.
Let $\mathcal C_{(n,\delta,b)}$ denote the cyclic code of length $n$ with generator
polynomial $g_{(n,\delta,b)}(x)$.
Then $\mathcal C_{(n,\delta,b)}$ is called a BCH code with \emph{designed distance} $\delta$,
and $S=\{uQ^\ell: b \le u \le b+\delta-2, 0 \le \ell \le m-1\}$ is the
\emph{defining set} of $\mathcal C_{(n,\delta,b)}$.
We call $\mathcal C_{(n,\delta,b)}$  a \emph{narrow-sense BCH code} if $b=1$ and \emph{non-narrow-sense BCH code} otherwise. When $n=Q^m-1$,
$\mathcal C_{(n,\delta,b)}$ is called  a \emph{primitive BCH code}. For more information on BCH codes, we refer the reader to \cite{Charpin, DDZ15, HP03, MS77}.

So far, the Euclidean LCD codes have been well investigated \cite{DLL}, \cite{DKOSS}, \cite{EY09}, \cite{Massey64}, \cite{Massey92}, \cite{MTQ}, \cite{ML86}, \cite{Sendr}, \cite{TH70}, \cite{YM94}, \cite{ZPS}, which include their properties and applications. It should be remarked that \cite{DLL} gave a well rounded treatment of Euclidean LCD cyclic codes over finite fields.
For Hermitian LCD codes, there are a few references to study them, though it was proved that they are asymptotically good \cite{GOS}.
 Boonniyoma and Jitman gave a sufficient and necessary condition on Hermitian LCD codes by employing their generator matrices \cite{BJ}.

Let $f(x)=f_tx^t+f_{t-1}x^{t-1}+\cdots+f_1x+f_0$ be a polynomial over $\gf(Q)$ with $f_t \ne 0$ and $f_0 \ne 0$. The reciprocal $f^*(x)$
of $f(x)$ is defined by
$$f^*(x)=f_0^{-1}x^tf(x^{-1}).$$
Denote
$$\bar{f}(x)=f_t^qx^t+f_{t-1}^qx^{t-1}+\cdots+f_1^qx+f_0^q.$$
It is easy to check that two operations $*$ and $\bar{}$ are commutative, i.e., $\overline{(f^*)}(x)=(\bar{f})^*(x)$.
Very recently, G\"{u}neri, \"{O}zkaya, and Sol\'{e} \cite{GOS} presented a nice description of cyclic Hermitian LCD  codes over finite fields.

\begin{lemma} \cite{GOS} \label{LHLCD} Let $\mathcal C$ be a cyclic code over $\gf(Q)$ with generator polynomial $g(x)$. Then $\mathcal C$ is a Hermitian LCD code if and only if $g(x)=\bar{g}^*(x)$.
\end{lemma}

The objective of this paper is to construct several classes of cyclic Hermitian LCD  codes over $\gf(Q)$ by employing cyclic codes and analyse their parameters. The dimensions of these codes are settled. The main difficulty in determining their dimensions is to use $q$-adic expansion to investigate the structure of $Q$-cyclotomic cosets, which is extremely complex. It is one of the major differences between
this paper and references \cite{DLL}, \cite{LDL} and \cite{LDLiu}. The other difference is that we extend some constructions in these papers. It might be true that the dimensions of cyclic Hermitian LCD  codes are not available in the literature. We remark that most Hermitian LCD codes presented in this paper are not BCH codes in general.
The lower bounds on minimum distances of the codes constructed in this paper are presented.
 In addition, we employ Hermitian LCD codes to propose a Hermitian orthogonal direct sum masking scheme that achieves protection against fault injection attacks.
 It will be shown that the Hermitian LCD codes with great minimum distances are desired, which can improve the resistance.

\section{$Q$-cyclotomic cosets modulo $n$}

Assume that $\gcd(n, Q)=1$. To deal with cyclic codes of length $n$ over $\gf(Q)$, we need to introduce $Q$-cyclotomic cosets modulo $n$ and the canonical factorization of $x^n-1$
over $\gf(Q)$.

Let $\Bbb Z_n = \{0,1,2, \cdots, n-1\}$ denote the ring of integers modulo $n$. For any $s \in \Z_n$, the \emph{$Q$-cyclotomic coset of $s$ modulo $n$\index{$Q$-cyclotomic coset modulo $n$}} is defined by
$$C_s=\{s, sQ, sQ^2, \cdots, sQ^{\ell_s-1}\} \bmod n \subseteq \Z_n,
$$
where $\ell_s$ is the smallest positive integer such that $s \equiv s Q^{\ell_s} \pmod{n}$, and is the size of the
$Q$-cyclotomic coset. The smallest integer in $C_s$ is called the \emph{coset leader\index{coset leader}} of $C_s$.
Let $\Gamma_{(n,Q)}$ be the set of all the coset leaders. We have then $C_s \cap C_t = \emptyset$ for any two
distinct elements $s$ and $t$ in  $\Gamma_{(n,Q)}$, and
\begin{eqnarray}\label{eqn-cosetPP}
\bigcup_{s \in  \Gamma_{(n,Q)} } C_s = \Z_n.
\end{eqnarray}
Hence, the distinct $Q$-cyclotomic cosets modulo $n$ partition $\Z_n$.

Let $m=\ord_{n}(Q)$, and let $\alpha$ be a generator of $\gf(Q^m)^*$. Put $\beta=\alpha^{(Q^m-1)/n}$.
Then $\beta$ is a primitive $n$-th root of unity in $\gf(Q^m)$. The minimal
polynomial $m_{s}(x)$ of $\beta^s$ over $\gf(Q)$ is the monic polynomial of the smallest degree over
$\gf(Q)$ with
$\beta^s$ as a zero. It is now straightforward to prove that this polynomial is given by
\begin{eqnarray*}
m_{s}(x)=\prod_{i \in C_s} (x-\beta^i) \in \gf(Q)[x],
\end{eqnarray*}
which is irreducible over $\gf(Q)$. It then follows from (\ref{eqn-cosetPP}) that
\begin{eqnarray*}\label{eqn-canonicalfact}
x^n-1=\prod_{s \in  \Gamma_{(n,Q)}} m_{s}(x)
\end{eqnarray*}
which is the factorization of $x^n-1$ into irreducible factors over $\gf(Q)$.

\section{Hermitian LCD codes from cyclic codes}

For any monic irreducible factor $f(x)$ of $x^n-1 \in \gf(Q)[x]$, $\bar{f}^*(x)$ is also a  monic irreducible factor of $x^n-1$.
It is easy to see that
\begin{equation} \label{EF} x^n-1=e_1(x)e_2(x)\cdots e_u(x)f_1(x)\bar{f}^*_1(x)f_2(x)\bar{f}^*_2(x)\cdots f_v(x)\bar{f}^*_v(x),\end{equation}
where $e_i(x)$ and $f_j(x)$ are monic irreducible polynomials over $\gf(Q)$, and $e_i(x)=\bar{e}^*_i(x)$ for $i=1, 2, \ldots, u$.
The following conclusion on the construction and the total number of the cyclic Hermitian LCD  codes of length $n$ then follows.

\begin{theorem} \label{TH}
 Let $\mathcal C$ be a cyclic code of length $n$ over $\gf(Q)$ with generator polynomial
 $$g(x)=e_1(x)^{a_1}e_2(x)^{a_2}\cdots e_u(x)^{a_u}f_1(x)^{b_1}\bar{f}^*_1(x)^{c_1}f_2(x)^{b_2}\bar{f}^*_2(x)^{c_2}\cdots f_v(x)^{b_v}\bar{f}^*_v(x)^{c_v},$$
 where $a_i, b_j,c_j \in \{0, 1\}$. Then $\mathcal C$ is a Hermitian LCD code if and only if $b_j=c_j$ for all $j=1, 2, \ldots, v$.
  Thus the total number of cyclic Hermitian LCD  codes of length $n$ is $2^{u+v}$ (Here $\{\textbf{0}\}$ and $\gf(Q)^n$ are also viewed
 as Hermitian LCD codes).
\end{theorem}

Theorem \ref{TH} indicates that the number of cyclic Hermitian LCD  codes of length $n$ is determined by \eqref{EF}.It also gives us a hint to
choose the defining set $S$ of the cyclic Hermitian LCD  codes, which should have the following type:
$$S=\{0 \le i \le n-1: g(\beta^i)=0\},$$
where $g(x)=\prod e_i(x) \prod f_j(x)\bar{f}^*_j(x)$ and $\beta$ is a primitive $n$-th root of unity.

\begin{theorem} \label{thm-DS} Let $\mathcal C$ be a cyclic code over $\gf(Q)$ with generator polynomial $g(x)$. Then $\mathcal C$ is a Hermitian LCD code if and only if one of the following statements holds:
\begin{enumerate}
  \item $S=-qS$ for the defining set $S$ of $\mathcal C$, where $-qS=\{-qs: s \in S\}$.
  \item $\beta^{-q}$ is a root of $g(x)$ for every root $\beta$ of $g(x)$.
\end{enumerate}
\end{theorem}

\begin{proof} It is straightforward from Theorem \ref{TH}. \end{proof}

\begin{theorem} \label{TO} Let $Q=q^2$. Then $-1$ is an odd power of $q$ modulo $n$ if and only if every cyclic code over $\gf(Q)$ of length $n$ is Hermitian LCD.

In particular, every cyclic code $\mathcal C$ of length $n=q^{2t+1}+1$ over $\gf(Q)$ is Hermitian LCD, where $t \ge 0$ is an integer.
\end{theorem}

\begin{proof}
Let $-1 \equiv q^{2t+1} \pmod n$ for some integer $t \ge 0$. Then for every integer $a$ with $0 \le a \le n-1$ we have
$$-a \equiv a q^{2t+1} \pmod n \text{ and } -qa \equiv a q^{2t+2} \equiv aQ^{t+1} \pmod n.$$
We then see that $-qa \in C_a$, i.e., $f_j(x)$ do not appear in \eqref{EF}.  By Theorem \ref{TH}, every cyclic code over $\gf(Q)$ of length $n$ is Hermitian LCD.

Let $\beta$ be a primitive $n$-th root of unity and $m_1(x)$ the minimal polynomial of $\beta$ over $\gf(Q)$. Suppose that $\mathcal C$ is a cyclic code over $\gf(Q)$ with
generator polynomial $m_1(x)$. Note that $\mathcal C$ is Hermitian LCD. Then
$$-q \in C_1 \text{ and } -q \equiv Q^t \pmod n$$ for some integer $t$ with $1 \le t \le m-1$. Note that $\gcd(n,q)=1$. Thus $-1 \equiv q^{2t-1} \pmod n$.
\end{proof}

The following lemma, which was given by Aly, Klappenecker, and Sarvepalli, will be employed later.

\begin{lemma} \cite{AKS} \label{lem-AKS}
Let $n$ be a positive integer such that $\gcd(n, Q)=1$ and $Q^{\lfloor m/2 \rfloor}<n \leq Q^m-1$, where
$m=\ord_n(Q)$. Then the $Q$-cyclotomic coset $C_s=\{sQ^j \bmod{n}: 0 \leq j \leq m-1\}$ has cardinality
$m$ for all $s$ in the range $1 \leq s \leq n Q^{\lceil m/2 \rceil}/(Q^m-1)$. In addition, every $s$ with
$s \not\equiv 0 \pmod{Q}$ in this range is a coset leader.
\end{lemma}

Denote $m=\ord_n(Q)$, where $n=q^{2t+1}$ and $Q=q^2$. It is easy to check that $m=2t+1$. By Lemma \ref{lem-AKS}, we see that
$i$ is the coset leader of $Q$-cyclotomic coset $C_i$ modulo $q^{2t+1}$, where $1 \le i \le q$.

\begin{theorem} \label{thm-HOP}
For $q=2$ and $Q=4$, let $\mathcal C_{(2^{2t+1}+1, 4, 0)}$ be the cyclic code with generator polynomial $g_{(2^{2t+1}+1, 4, 0)}$ given by
\eqref{GPoly}. Then $\mathcal C_{(2^{2t+1}+1, 4, 0)}$ is a quaternary Hermitian LCD code with parameters
$$[2^{2t+1}+1, 2^{2t+1}-4t-2, d \ge 6].$$
\end{theorem}

\begin{proof}
One can easily see that $\mathcal C_{(2^{2t+1}+1, 4, 0)}$ is a Hermitian LCD code and its defining set is
$$S=C_0 \cup C_1 \cup C_2,$$
which is a disjoint union. By Lemma \ref{lem-AKS}, we have $|C_1|=|C_2|=2t+1$ and $|S|=4t+3$. Then the dimension of
$\mathcal C_{(2^{2t+1}+1, 4, 0)}$ follows. Notice that $\{n-2, n-1, 0, 1, 2\}$ is a subset of $S$ from Theorem \ref{thm-DS}.
The lower bound of minimum distance then follows form the BCH bound.
\end{proof}

The tables of
best known linear codes are referred to as the \emph{Database} later, which are maintained by Markus Grassl at http://www.codetables.de/.

\begin{example} \label{exa-1}
Let $q=2$ and $Q=4$. Let $\mathcal C_{(2^{2t+1}+1, 4, 0)}$ be the quaternary Hermitian LCD code presented in Theorem \ref{thm-HOP}.
\begin{enumerate}
  \item When $t=1$, the quaternary code $\mathcal C_{(9, 4, 0)}$ has parameters $[9, 2, 6]$, which is almost optimal according to the Database.
  \item When $t=2$, the quaternary code $\mathcal C_{(33, 4, 0)}$ has parameters $[33, 22, 6]$, which is almost optimal according to the Database.
  \item When $t=3$, the quaternary code $\mathcal C_{(129, 4, 0)}$ has parameters $[129, 114, 6]$, which is an optimal code according to the Database.
\end{enumerate}
\end{example}

\section{cyclic Hermitian LCD  codes of length $Q^m-1$}

In this section, we always assume that $n=Q^m-1$, $\alpha$ is a generator of $\gf(Q^m)^*$, and $m_i(x)$ is the minimal polynomial of
$\alpha^i$ over $\gf(Q)$. Below we construct a class of cyclic Hermitian LCD  codes via BCH codes and investigate their fundamental parameters.

Let $e$ be a divisor of $q+1$, i.e., $e \mid (q+1) \mid n$. Write $\hat n=\frac n e$. Then $\alpha^{\hat n} \in \gf(Q)$.
For any integer $\delta$ with $2 \leq \delta \leq n$, denote
\begin{equation} \label{G1} g(x)=(x-\alpha^{\hat n})g_{(n, \delta,{\hat n}+1)}(x)\lcm\big(m_{{\hat n}-q(\delta-1)}(x),m_{{\hat n}-q(\delta-2)}(x),\ldots, m_{{\hat n}-q}(x)\big) \in \gf(Q)(x),\end{equation}
where $g_{(n, \delta,\hat n+1)}(x)$ is defined by \eqref{GPoly}. Let $\mathcal C$ be a cyclic code of length $n$ with generator
polynomial $g(x)$. One can see that the defining set of $\mathcal C$ is
$$S=\bigcup_{i=1}^{\delta-1}(C_{\hat n+i} \cup C_{{\hat n}-qi}) \bigcup C_{\hat n}.$$
It is clear that
$$-q(\hat n+i) \equiv -q\hat n-qi \equiv \hat n-qi \pmod n$$ and
$$-q(\hat n-qi) \equiv \hat n+q^2 i \equiv q^2 \hat n +q^2i=Q(\hat n+i) \pmod n.$$
It then follows that $S=-qS$. By Theorem \ref{thm-DS}, we conclude that $\mathcal C$ is a cyclic Hermitian LCD  code over $\gf(Q)$.
It should be noticed that $C$ may be not a BCH code in general.

\emph{A. Cyclotomic cosets modulo  $Q^m-1$}

To present the dimension of the code $\mathcal C$, we need some results on $Q$-cyclotomic cosets.

\begin{lemma} \label{LPCL} \cite{AKS, YF} Let $m \ge 2$ be an integer and denote $\bar m=\lceil \frac m 2 \rceil$. For any $i$ with $1 \le i \le Q^{\bar m}$ and $Q \nmid i$, $i$ is a coset leader of the cyclotomic coset $C_i$ and
$|C_i|=m$ for all $i$ in the range $1 \le i \le Q^{\bar m}$.
\end{lemma}

The following lemma on $Q$-cyclotomic cosets modulo $Q^m-1$ will be employed later.

\begin{lemma} \label{LCC} Let $n=Q^m-1$ and $\hat n=\frac n e$. Then we have the following.
\begin{enumerate}
  \item $|C_{\hat n +i}|=|C_{\hat n-qi}|=|C_i|$.
  \item $C_{\hat n +i}=C_{\hat n +j}$ if and only if $C_i=C_j$.
 \item $C_{\hat n -qi}=C_{\hat n-qj}$ if and only if $C_i=C_j$.
 \item $C_{\hat n+Qi}=C_{\hat n+i}$ and $C_{\hat n -qQi}=C_{\hat n -qi}$.
\end{enumerate}
\end{lemma}

\begin{proof}
Let $\ell$ be a positive integer. Note that $e \mid Q^\ell-1$. It is clear that
$$(\hat n+i)Q^\ell \equiv \hat n+j \pmod n$$ is equivalent to
$$Q^\ell i \equiv j \pmod n.$$ The desired conclusions then follow.
\end{proof}

For $n=Q^m-1$, denote
$$J^+(\delta)=\bigcup_{i=1}^{\delta-1}C_{\hat n+i} \text{ and } J^-(\delta)=\bigcup_{i=1}^{\delta-1}C_{\hat n-qi}.$$
The following proposition will play an important role in presenting the dimension of the code $\mathcal C$ when $\delta \le Q^{\bar m}+1$.

\begin{proposition} \label{PPrimitiveCase} Let $m \ge 2$, $\bar m=\lceil \frac m 2 \rceil$ and $2 \le \delta \le Q^{\bar m}+1$. Then the following hold.
\begin{enumerate}
  \item $|J^+(\delta)|=|J^-(\delta)|=(\delta-1-\lfloor \frac {\delta-1} Q \rfloor)m.$
  \item When $m$ is odd, for $1 \le u \le q-1$, we have
  $$|J^+(\delta)\bigcap J^-(\delta)|=\begin{cases}
  0, & \text{ if } 2 \le \delta \le q^m-1; \\
  u^2m, & \text{ if } uq^m \le \delta \le (u+1)(q^m-1); \\
 (u^2+2v+1)m, & \text{ if } \delta=(u+1)(q^m-1)+v+1 \text{ for } 0 \le v \le u-1; \\
  q^2 m, & \text{ if } \delta=q^{m+1} \text{ or } q^{m+1}+1.
  \end{cases}$$
  \item  When $m$ is even, for any $\delta$ with $2 \le \delta \le q^m+1$, we have
  $$|J^+(\delta)\bigcap J^-(\delta)|=0.$$
\end{enumerate}
\end{proposition}

\begin{proof} It is easily deduced from Lemma \ref{LCC} that
$C_{\hat n +i} \ne C_{\hat n +j}$ if and only if $C_i \ne C_j$.  Note that every integer $i$ with $1 \le i \le Q^{\bar m}$ and $Q \nmid i$ is a coset leader of the cyclotomic coset $C_i$ by Lemma \ref{LPCL}. Then
$$|J^+(\delta)|=(\delta-1-\lfloor \frac {\delta-1} Q \rfloor)m.$$
Similarly, we have
$$|J^-(\delta)|=(\delta-1-\lfloor \frac {\delta-1} Q \rfloor)m.$$

For the case that $m$ is odd, we have $\bar m=\frac {m+1} 2$.
Suppose that $a \in J^+(\delta) \cap J^-(\delta)$. Then there exist $i$ and $-qj$ with $1 \le i, j \le \delta-1$ such that
$$a \in C_{\hat n +i}=C_{\hat n-qj},$$
which implies that
\begin{equation}\label{EIJ1} i \equiv -qjQ^\ell \pmod {Q^m-1} \text{  and } i+jq^{2\ell+1} \equiv 0 \pmod {q^{2m}-1}\end{equation} for
some $1 \le \ell \le m-1$.

We can further assume that $Q \nmid i$ and $Q \nmid j$. Note that $i, j \le Q^{\bar m}=q^{m+1}$. Then we have the following two $q$-adic expansions
   $$i=i_{m}q^{m}+i_{m-1}q^{m-1}+\cdots+i_1q+i_0$$ and
$$j=j_{m}q^{m}+j_{m-1}q^{m-1}+\cdots+j_1q+j_0,$$
where $0 \le i_k, j_k \le q-1$ for all $k$ with $0 \le k \le m$, $(i_1, i_0) \ne (0,0)$, and $(j_1,j_0) \ne (0,0)$.

  \emph{Case 1:} When $1 \le \ell \le \frac {m-3} 2$, it is easy to check that $0 < i+jq^{2\ell+1} < q^{2m}-1$. It then follows that $i+jq^{2\ell+1} \equiv 0 \pmod {q^{2m}-1}$ doesn't hold.

   \emph{Case 2:} When $\ell=\frac {m-1} 2$, it can be verified that $i+jq^{2\ell+1} \equiv \Delta \pmod n$, where
$$\Delta=j_{m-1}q^{2m-1}+\cdots+j_1q^{m+1}+(j_0+i_{m})q^{m}+i_{m-1}q^{m-1}+\cdots+i_1q+(i_0+j_{m}).$$
Notice that $0< \Delta < 2n$. It then follows from \eqref{EIJ1} that $\Delta=n$. Thus
$$j_{m-1}=\cdots=j_1=j_0+i_{m}=i_{m-1}=\cdots=i_1=i_0+j_{m}=q-1.$$
Then $$i_m=u, \ \ j_m=v, \ \ i_0=q-1-v, \ \ j_0=q-1-u,$$ where $u, v=0, 1, 2, \ldots, q-1$.
 Hence $$i=(u+1)q^m-v-1 \text{ and } j=(v+1)q^m-u-1.$$

   \emph{Case 3:} When $\frac {m+1} 2 \le \ell \le m-1$, we have $m+2 \le 2\ell+1 \le 2m-1$. Denote $2\ell+1=m+\epsilon$, where $2 \le \epsilon \le m-1$. Then
  $i+jq^\ell \equiv \Delta \pmod n$, where
  \begin{eqnarray*}\Delta &=&j_{m-\epsilon-1}q^{2m-1}+\cdots+j_1q^{m+\epsilon+1}+j_0q^{m+\epsilon}+i_mq^m+\cdots+i_{\epsilon+1}q^{\epsilon+1}
  +(i_{\epsilon}+j_m)q^{\epsilon}\\&+&\cdots+(i_1+j_{m-\epsilon+1})q+(i_0+j_{m-\epsilon}). \end{eqnarray*}
   It is easy to see that the coefficient of $q^{m-1}$ in the $q$-adic expansion of $\Delta$ is equal to $0$.
 Thus we have $0<\Delta <n$, which means that \eqref{EIJ1} is impossible.

Denote $i=i_{uv}$ and $j=j_{uv}$ described in Case 3. Write
$$L=\{(u,v): i_{uv} \le \delta-1, j_{uv} \le \delta-1\}.$$ Note that Cases 1, 2, and 3 contain all possible pairs $(i, j)$ such that
$C_{\hat n +i}=C_{\hat n-qj}$ with $1 \le i, j \le \delta-1$, where $\delta \le q^{m+1}+1$. Thus
\begin{equation} \label{DU} J^+(\delta)\bigcap J^-(\delta)=\bigcup_{(u, v)\in L}C_{\hat n+i_{uv}}.\end{equation}
We have seen that each $i_{uv}$ is a coset leader from Lemma \ref{LPCL}. Also, by Lemma \ref{LCC},
$C_{\hat n +i} \ne C_{\hat n +j}$ if and only if $C_i \ne C_j$. Hence, the union in \eqref{DU} is disjoint. It then
follows that $|J^+(\delta)\bigcap J^-(\delta)|=|L|m$.

Now we are going to employ two auxiliary matrices to evaluate $|L|$.
Put $i_{uv}$ and $j_{uv}$ into two matrices $A$ and $B$, respectively, where
$$A=\big(i_{uv}\big)=\left(\begin{array}{ccccccc}
      q^m-1 & q^m-2 & q^m-3 & \cdots & q^m-u & \cdots & q^m-q \\
      2q^m-1 & 2q^m-2 & 2q^m-3 & \cdots & 2q^m-u & \cdots & 2q^m-q \\
      3q^m-1 & 3q^m-2 & 3q^m-3 & \cdots & 3q^m-u & \cdots & 3q^m-q \\
      \vdots & \vdots & \vdots & \ddots & \vdots & \vdots & \vdots \\
       uq^m-1 & uq^m-2 & uq^m-3 & \cdots & uq^m-u & \cdots & uq^m-q \\
       \vdots & \vdots & \vdots & \ddots & \vdots & \ddots & \vdots \\
      q^{m+1}-1 & q^{m+1}-2 & q^{m+1}-3 & \cdots & q^{m+1}-u & \cdots & q^{m+1}-q
    \end{array}\right)_{q \times q}$$  and
 $$B=\big(j_{uv}\big)=\left(\begin{array}{ccccccc}
      q^m-1 & 2q^m-1 & 3q^m-1 & \cdots & uq^m-1 & \cdots & q^{m+1}-1 \\
      q^m-2 & 2q^m-2 & 3q^m-2 & \cdots & uq^m-2 & \cdots & q^{m+1}-2 \\
      q^m-3 & 2q^m-3 & 3q^m-3 & \cdots & uq^m-3 & \cdots & q^{m+1}-3 \\
      \vdots & \vdots & \vdots & \ddots & \vdots & \vdots & \vdots \\
       q^m-u & 2q^m-u & 3q^m-u & \cdots & uq^m-u & \cdots & q^{m+1}-u \\
       \vdots & \vdots & \vdots & \ddots & \vdots & \ddots & \vdots \\
      q^{m}-q & 2q^{m}-q & 3q^{m}-q & \cdots & uq^{m}-q & \cdots & q^{m+1}-q
    \end{array}\right)_{q \times q}.$$   In fact, $B$ is the transpose of $A$. With the help of the patterns of $A$ and $B$, it is not hard to see  that
    $$|L|=\begin{cases}
  0, & \text{ if } 2 \le \delta \le q^m-1; \\
  u^2, & \text{ if } uq^m \le \delta \le (u+1)(q^m-1)\text{ for } 1 \le u \le q-1; \\
  u^2+2v+1, & \text{ if } \delta=(u+1)(q^m-1)+v+1 \text{ for } 0 \le v \le u-1; \\
  q^2, & \text{ if } \delta=q^{m+1} \text{ or } q^{m+1}+1.
  \end{cases}$$
    For example, if $\delta-1=u(q^m-1)$, i.e., $\delta=u(q^m-1)+1$, then one can see that
    $$L=\{(s, t): 1 \le s, t \le u-1\} \cup \{(u, u)\}$$ and $|L|=(u-1)^2+1$.
Hence we get the desired conclusion when $m$ is odd.
The case that $m$ is even can be similarly dealt with and the proof is omitted here.
\end{proof}

\emph{B. Parameters of the Hermitian LCD code $\mathcal C$}

\begin{theorem} \label{T1} Let $m \ge 2$, $\bar m=\lceil \frac m 2 \rceil$, $Q=q^2$, and $2 \le \delta \le Q^{\bar m}+1$.
Let $\mathcal C$ be a cyclic code of length $n$ with generator
polynomial $g(x)$ given by \eqref{G1}. Then $\mathcal C$ is a cyclic Hermitian LCD  code over $\gf(Q)$ and the following hold.
\begin{enumerate}
  \item When $m$ is odd,  for $1 \le u \le q-1$, the code $\mathcal C$ has length
$n=Q^m-1$, dimension $k$, and minimum distance $d \ge \delta+1+\lfloor \frac {\delta-1} q \rfloor$, where
$$k=\begin{cases}
  Q^m-2-2(\delta-1-\lfloor \frac {\delta-1} Q\rfloor)m, & \text{ if } 2 \le \delta \le q^m-1; \\
  Q^m-2-2(\delta-1-\lfloor \frac {\delta-1} Q\rfloor)m+u^2m, & \text{ if } uq^m \le \delta \le (u+1)(q^m-1); \\
  Q^m-2-2(\delta-1-\lfloor \frac {\delta-1} Q\rfloor)m+(u^2+2v+1)m, & \text{ if } \delta=(u+1)(q^m-1)+v+1 \text{ for } 0 \le v \le u-1;  \\
  Q^m-2-2(\delta-1-\lfloor \frac {\delta-1} Q\rfloor)m+q^2m, & \text{ if } \delta=q^{m+1} \text{ or } q^{m+1}+1.
  \end{cases}$$
  \item When $m$ is even, the code $\mathcal C$ has length
$n=Q^m-1$, dimension $$k=Q^m-2-2(\delta-1-\lfloor \frac {\delta-1} Q\rfloor)m,$$ and minimum distance $d \ge \delta+1+\lfloor \frac {\delta-1} q \rfloor$.
\end{enumerate}
\end{theorem}

\begin{proof}
By definition, the degree of the generator polynomial $g(x)$ of $\mathcal C$ is equal to $$1+|J^+(\delta)|+|J^-(\delta)|-|J^+(\delta)\bigcap J^-(\delta)|.$$
The dimension $k=Q^m-1-\deg(g(x))$ then follows from Proposition \ref{PPrimitiveCase}. Note that $C_{\hat n-qi}=C_{\hat n-i/q}$ if $q \mid i$. Then the set
$$\{\hat n-\lfloor \frac {\delta-1} q \rfloor,\hat n -\lfloor \frac {\delta-1} q \rfloor+1, \ldots, \hat n-1, \hat n, \hat n+1, \hat n+2,\ldots, \hat n+\delta-1\}$$ is
contained in the defining set of $\mathcal C$. The desired conclusion on minimum distance then follows from the BCH bound.
\end{proof}

\section{Quaternary cyclic Hermitian LCD  codes of length $\frac {4^m-1} 3$}

In this section, we always assume that $q=2$, $Q=4$, and $n=\frac {4^m-1} 3$.
For any integer $\delta$ with $2 \leq \delta \leq n$, denote
\begin{equation} \label{G2} g(x)=(x-1)g_{(n, \delta,1)}(x)\lcm\big(m_{n-2(\delta-1)}(x),m_{n-2(\delta-2)}(x),\ldots, m_{n-2}(x)\big) \in \gf(4)(x),\end{equation}
where $g_{(n, \delta,1)}(x)$ is defined by \eqref{GPoly}. Let $\mathcal C$ be a cyclic code of length $n$ with generator
polynomial $g(x)$. It is clear that the defining set of $\mathcal C$ is
$$S=\bigcup_{i=1}^{\delta-1}(C_{i} \cup C_{n-2i}) \bigcup C_0.$$
Note that $S=-2S$. By Theorem \ref{thm-DS}, it can be concluded that $\mathcal C$ is a cyclic Hermitian LCD  code over $\gf(4)$.
In this section, we investigate the parameters of the quaternary cyclic Hermitian LCD  code $\mathcal C$ when $\delta \le 2^m$.
We first give some preliminaries on $4$-cyclotomic cosets modulo $\frac {4^m-1} 3$.

\emph{A. Cyclotomic cosets modulo  $\frac {4^m-1} 3$}

When $m$ is even, the following lemma characterizes all coset leaders in the range $1 \le i \le 2^m$.

\begin{lemma} \label{LME}
Let $m \ge 2$ be an even integer. For any $i$ with $1 \le i \le 2^m$ and $4 \nmid i$, then
\begin{enumerate}
  \item $|C_i|=m$;
  \item $i$ is a coset leader of the $4$-cyclotomic coset $C_i$  with $|C_i|=m$ except that $i=\frac {2^{m+1}+1} 3$.
\end{enumerate}
\end{lemma}

\begin{proof}
Denote $\bar m=\frac m 2$. For $i$ in the range $1 \le i \le 4^{\bar m}$, we have the $4$-adic expansion
$$i=i_{\bar m-1}4^{\bar m-1}+i_{\bar m-2}4^{\bar m-2}+\cdots+i_1 4+i_0,$$
where $0 \le i_t \le 3$ for $1 \le t \le \bar m-1$ and $i_0 \ne 0$.

For an integer $\ell$ with $1 \le \ell \le m-1$, we have
$$i4^\ell=i_{\bar m-1}4^{\bar m+\ell-1}+i_{\bar m-2}4^{\bar m+\ell-2}+\cdots+i_1 4^{\ell+1}+i_04^\ell,$$
\begin{enumerate}
  \item When $\ell \le \bar m-1$, one can check that $i < i4^\ell < n$.
  \item When $\ell=\bar m$, $\bar m+\ell-1=m-1$. It then follows that
  \begin{eqnarray*}i4^\ell &\equiv& i_{\bar m-2} 4^{m-2}+\cdots+i_14^{\ell+1}+i_04^\ell-i_{\bar m-1}(4^{m-2}+\cdots +4+1) \\
  &\equiv& (i_{\bar m-2}-i_{\bar m-1})4^{m-2}+\cdots+(i_0-i_{\bar m-1})4^\ell-i_{\bar m-1}(4^{\ell-1}+\cdots+4+1) \pmod n.\end{eqnarray*}

 Let $\overline i$ denote the integer in the range $0 \le \overline i \le n-1$ with $i \equiv \overline i \pmod n$.
 If $i_{\bar m-2}-i_{\bar m-1}=\cdots=i_0-i_{\bar m-1}=0$, then $$i < \overline{i4^\ell}=n-i_{\bar m-1}(4^{\ell-1}+\cdots+4+1)<n.$$
 Otherwise, let $k$ be the largest integer such that $i_k-i_{\bar m-1} \ne 0$ with $0 \le k \le \bar m-2$. Then
  $i4^\ell \equiv \Delta \pmod n$, where
  $$\Delta=(i_k-i_{\bar m-1})4^{k+\ell}+\cdots+(i_0-i_{\bar m-1})4^\ell-i_{\bar m-1}(4^{\ell-1}+\cdots+4+1).$$
  If $i_k-i_{\bar m-1} < 0$, then $$i < \overline{i4^\ell}=n-\Delta < n.$$ If $i_k-i_{\bar m-1} > 0$ and $k \ge 1$, then $$i < \overline{i4^\ell}=\Delta < n.$$

  We then consider the case $i_0-i_{\bar m-1} > 0$, i.e., $k=0$. One can see that
  $$\overline{i4^\ell}-i=(i_0-i_{\bar m-1})4^{\bar m}-i_{\bar m-1}(4^{\bar m-1}+\cdots+4+1)-i.$$
  When $i_0-i_{\bar m-1} \ge 2$, it is easy to check that
  $$\overline{i4^\ell}-i >0.$$
  When $i_0-i_{\bar m-1}=1$, we have $i_0=i_{\bar m-1}+1$. Note that $i_{\bar m-1}=i_{\bar m-2}=\cdots=i_1$. Then
  $$\overline{i4^\ell}-i=4^{\bar m}-2i_{\bar m-1}\frac {4^{\bar m}-1} 3-1.$$
  It then can be verified that
  $\overline{i4^\ell}-i<0$ if and only if
  $$i_{\bar m-1}=i_{\bar m-2}=\cdots=i_1=2 \text{ and } i_0=3, \text{ i.e. }, i=\frac {2^{m+1}+1} 3.$$
  \item When $\ell \ge \bar m+1$, one can similarly check that $i \le \overline{i4^\ell} <n$ and we omit the details.
\end{enumerate}
Summarizing all the discussions above, we get the desired result.
\end{proof}

When $m$ is odd, all coset leaders in the range $1 \le i \le 2^m$ are characterized in the following lemma.

\begin{lemma} \label{LMO}
Let $m \ge 5$ be an odd integer. For any $i$ with $1 \le i \le 2^m$ and $4 \nmid i$, then
 \begin{enumerate}
   \item $|C_i|=m$;
   \item $i$ is a coset leader of the $4$-cyclotomic coset $C_i$  with $|C_i|=m$ except that $i=\frac {2^{m+1}+2} 3$ and $i=\frac {2^{m+1}+2^{m-1}+1} 3$.
 \end{enumerate} 
\end{lemma}

\begin{proof} The proof is similar to that of Lemma \ref{LME} and we give a sketch here.
Denote $\bar m=\frac {m+1} 2$. For $i$ in the range $1 \le i \le 2\cdot4^{\bar m-1}-1$, we have the $4$-adic expansion
$$i=i_{\bar m-1}4^{\bar m-1}+i_{\bar m-2}4^{\bar m-2}+\cdots+i_1 4+i_0,$$
where $i_{\bar m-1}=0, 1$, $0 \le i_t \le 3$ for $1 \le t \le \bar m-2$, and $i_0 \ne 0$.

For an integer $\ell$ with $1 \le \ell \le m-1$, we have
$$i4^\ell=i_{\bar m-1}4^{\bar m+\ell-1}+i_{\bar m-2}4^{\bar m+\ell-2}+\cdots+i_1 4^{\ell+1}+i_04^\ell,$$
\begin{enumerate}
  \item When $\ell \le \bar m-2$, it is easy to check that $i < i4^\ell < n$.
  \item When $\ell=\bar m-1$, one can prove that
  $\overline{i4^\ell} < i$ if and only if
  $$i_{\bar m-1}=i_{\bar m-2}=\cdots=i_1=1 \text{ and } i_0=2, \text{ i.e., } i=\frac {2^{m+1}+2} 3.$$
  \item When $\ell=\bar m$, it can be checked that
  $\overline{i4^\ell} < i$ if and only if
  $$i_{\bar m-1}=1, \ \  i_{\bar m-2}=i_{\bar m-3}=\cdots=i_1=2, \  \text{ and } i_0=3, \text{ i.e., } i=\frac {2^{m+1}+2^{m-1}+2} 3.$$
  \item When $\ell \ge \bar m+1$, it can be shown that
  $i < \overline{i4^\ell} < n$.
  \end{enumerate}
 Summarizing all the discussions above, the desired conclusion then follows.
\end{proof}

For $n=\frac {4^m-1} 3$ and $\delta \ge 2$, denote
$$J^+(\delta)=\bigcup_{i=1}^{\delta-1}C_i \text{ and } J^-(\delta)=\bigcup_{i=1}^{\delta-1}C_{-2i}.$$

\begin{proposition} \label{PME} When $m \ge 2$ is even and $\delta \le 2^m$, we have
$$|J^+(\delta) \bigcap J^-(\delta)|=\begin{cases}
0, & \text{ if } 2 \le \delta \le \frac {2^{m+1}-2} 3; \\
2m, & \text{ if } \frac {2^{m+1}-2} 3+1 \le \delta \le \frac {2^{m+1}+2^{m-1}-1} 3; \\
4m, & \text{ if } \frac {2^{m+1}+2^{m-1}-1} 3+1 \le \delta \le 2^m. \end{cases}$$
\end{proposition}

\begin{proof}
Suppose that $$a \in J^+(\delta) \bigcap J^-(\delta).$$ Then $a \in C_i=C_{-2j}$ for $1 \le i, j \le 2^m-1$.
Thus there is an integer $\ell$ with $0 \le \ell \le m-1$ such that $i \equiv -2j4^\ell \pmod n$ and
\begin{equation} \label{IJE} i+j2^{2\ell+1} \equiv 0 \pmod n. \end{equation}

For $1 \le i, j \le 2^m-1$ and $4 \nmid i, j$, we have the following two $2$-adic expansions
$$i=i_{m-1}2^{m-1}+i_{m-2}2^{m-2}+\cdots+i_1 2+i_0$$
and
$$j=j_{m-1}2^{m-1}+j_{m-2}2^{m-2}+\cdots+j_1 2+j_0,$$
where $i_t,j_t=0 \text{ or } 1$, $(i_1, i_0) \ne (0,0)$, and $(j_1, j_0) \ne (0,0)$.

\emph{Case 1:} When $\ell \le \frac m 2-2$, we have $m+2\ell \le 2m-4$ and $$i+j2^{2\ell+1} <n,$$
which implies that \eqref{IJE} does not hold.

\emph{Case 2:} When $\ell=\frac m 2 -1$, it can be shown that
$$i+j2^{2\ell+1} \equiv j_{m-1}2^{2m-2}+\cdots+j_12^m+(j_0+i_{m-1})2^{m-1}+\cdots+i_1 2+i_0 \pmod n.$$
Recall that $n=2^{2m-2}+2^{2m-4}+\cdots+2^2+1$. Since $m-1$ is odd, we have
$$j_0+i_{m-1}=0 \text{ or } j_0+i_{m-1}=2.$$

If $j_0+i_{m-1}=0$, then $j_0=i_{m-1}=0$. Note that $(i_1, i_0) \ne(0, 0)$. It then follows that
$$i_1=i_3=\cdots=i_{m-3}=0, \ \ i_0=i_2=\cdots=i_{m-2}=1,$$
$$j_1=j_3=\cdots=j_{m-1}=1, \ \  j_2=j_4=\cdots=j_{m-2}=0.$$
Thus
$$i=\frac {2^m-1} 3 \text{ and } j=\frac {2^{m+1}-2} 3.$$

If $j_0+i_{m-1}=2$, then $j_0=i_{m-1}=1$ and
$$i+j2^{2\ell+1} \equiv j_{m-1}2^{2m-2}+\cdots+j_2 2^{m+1}+(j_1+1)2^m+i_{m-2}2^{m-2}+\cdots+i_1 2+i_0 \pmod n.$$
We similarly have
$$i_1=i_3=\cdots=i_{m-3}=1, \ \ i_0=i_2=\cdots=i_{m-2}=0,$$
$$j_1+1=j_3=\cdots=j_{m-1}=1, \ \  j_2=j_4=\cdots=j_{m-2}=0.$$
Then
$$i=\frac {2^{m+1}+2^{m-1}-1} 3 \text{ and } j=\frac {2^{m+1}-5} 3.$$

\emph{Case 3:} When $\ell=\frac m 2$, one can get $i+j2^{2\ell+1} \equiv \Delta \pmod n$, where
\begin{eqnarray*}\Delta &=& j_{m-3}2^{2m-2}+(j_{m-4}-j_{m-2})2^{2m-3}+j_{m-5}2^{2m-4}+(j_{m-6}-j_{m-2})2^{2m-5}+\\& &  \cdots +j_12^{m+2}+(j_0-j_{m-2})2^{m+1}+
(i_{m-1}-j_{m-2})2^{m-1}+i_{m-2}2^{m-2}+(i_{m-3}-j_{m-2})2^{m-3}\\&&+i_{m-4}2^{m-4}+ \cdots+i_2 2^2 +(i_1-j_{m-2})2 +(i_0+j_{m-1}).\end{eqnarray*}
Note that
$$(i_{m-1}-j_{m-2})2^{m-1}+i_{m-2}2^{m-2}+(i_{m-3}-j_{m-2})2^{m-3}+i_{m-4}2^{m-4}+ \cdots+i_2 2^2 +(i_1-j_{m-2})2 +(i_0+j_{m-1}) \le 2^m.$$
It then follows that $\Delta=0$. Thus $i_0+j_{m-1} \ne 1$, i.e., $i_0+j_{m-1}=0 \text{ or } 2$. It also can be see that
$$j_{m-3}2^{2m-2}+(j_{m-4}-j_{m-2})2^{2m-3}+j_{m-5}2^{2m-4}+(j_{m-6}-j_{m-2})2^{2m-5}+ \cdots +j_12^{m+2}+(j_0-j_{m-2})2^{m+1}=0$$ and
\begin{equation} \label{JE} j_{m-3}=j_{m-4}-j_{m-2}=j_{m-5}=j_{m-6}-j_{m-2}=\cdots=j_1=j_0-j_{m-2}=0.\end{equation}

If $i_0+j_{m-1}=0$, then $i_0=j_{m-1}=0$, so $i_1=1$.
By $\Delta=0$, we have
$$i_2=i_4=\cdots=i_{m-2}=0, \ \ i_1-j_{m-2}=i_3-j_{m-2}=\cdots=i_{m-1}-j_{m-2}=0.$$
One can see $j_0=1$ by $j_1=0$. It then follows that
$$i_1=i_3=\cdots=i_{m-1}=1 \text{ and } j_0=j_2=\cdots=j_{m-4}=j_{m-2}=1.$$
Thus
$$i=\frac {2^{m+1}-2} 3 \text{ and } j=\frac {2^m-1} 3.$$

If $i_0+j_{m-1}=2$, then $i_0=j_{m-1}=1$. We claim that $j_{m-2}=1$. Otherwise, $j_0=j_1=0$ by \eqref{JE}, which is a contradiction.
We then have $i_1-j_{m-2} \le 0$. It can be deduced that
$$i_1-j_{m-2}+1=0, \ \  i_{m-2}=i_{m-4}=\cdots=i_2=0,$$
$$i_{m-1}-j_{m-2}=i_{m-3}-j_{m-2}=\cdots=j_0-j_{m-2}=0.$$
Thus
$$i_1=0, \ \  i_3=i_5=\cdots=i_{m-1}=1, \ \text{ and } j_0=j_2=\cdots=j_{m-2}=1.$$
We conclude that
$$i=\frac {2^{m+1}-5} 3 \text{ and } j=\frac {2^{m+1}+2^{m-1}-1} 3.$$

\emph{Case 4:} When $\ell> \frac m 2$, let $\ell=\frac m 2+\epsilon$, where $1 \le \epsilon \le \frac m 2-1$.
One can verify that $i+j2^{2\ell+1} \equiv \Delta \pmod n$, where
\begin{eqnarray*}\Delta&=&j_{m-2\epsilon-3}2^{2m-2}+\cdots+j_1 2^{m+2\epsilon+2}+j_02^{m+2\epsilon+1}+i_{m-1}2^{m-1}+\\ && \cdots+i_{2\epsilon+2}2^{2\epsilon+2}+
(i_{2\epsilon+1}+j_{m-1})2^{2\epsilon+1}+\cdots+(i_0+j_{m-2\epsilon-2}) < n.\end{eqnarray*}
It then follows that \eqref{IJE} does not hold.

Observe that Cases 1, 2, 3, and 4 contain all possible pairs $(i, j)$ such that
$a \in C_i=C_{-2j}$ with $1 \le i, j \le 2^m-1$. Thus,
$$J^+(\delta) \bigcap J^-(\delta)=\begin{cases}
\emptyset, & \text{ if } 2 \le \delta \le \frac {2^{m+1}-2} 3; \\
C_{\frac {2^m-1} 3} \cup C_{\frac {2^{m+1}-2} 3}, & \text{ if } \frac {2^{m+1}-2} 3+1 \le \delta \le \frac {2^{m+1}+2^{m-1}-1} 3, \\
C_{\frac {2^m-1} 3} \cup C_{\frac {2^{m+1}-5} 3} \cup C_{\frac {2^{m+1}-2} 3} \cup C_{\frac {2^{m+1}+2^{m-1}-1} 3}, & \text{ if } \frac {2^{m+1}+2^{m-1}-1} 3+1 \le \delta \le 2^m, \end{cases}$$
where the unions are disjoint. The desired conclusion is then straightforward from Lemma \ref{LME}.
\end{proof}

\begin{proposition} \label{PMO} When $m \ge 5$ is odd and $\delta \le 2^m$, we have
$$|J^+(\delta) \bigcap J^-(\delta)|=\begin{cases}
0, & \text{ if } 2 \le \delta \le \frac {2^{m+1}-1} 3; \\
2m, & \text{ if } \frac {2^{m+1}-1} 3+1 \le \delta \le 2^m-1; \\
3m, & \text{ if } \delta=2^m. \end{cases}$$
\end{proposition}

\begin{proof}
Let $$a \in J^+(\delta) \bigcap J^-(\delta).$$ Then there are integers $i$ and $j$ with $1 \le i, j \le 2^m-1$ such that $a \in C_i=C_{-qj}$.
Thus $i \equiv -2j4^\ell \pmod n$ and
\begin{equation} \label{IJO} i+j2^{2\ell+1} \equiv 0 \pmod n \end{equation} for some integer $\ell$ with $0 \le \ell \le m-1$

For $1 \le i, j \le 2^m-1$ and $4 \nmid i, j$, we have the following two $2$-adic expansions
$$i=i_{m-1}2^{m-1}+i_{m-2}2^{m-2}+\cdots+i_1 2+i_0$$
and
$$j=j_{m-1}2^{m-1}+j_{m-2}2^{m-2}+\cdots+j_1 2+j_0,$$
where $i_t,j_t=0 \text{ or } 1$, $(i_1, i_0) \ne (0,0)$, and $(j_1, j_0) \ne (0,0)$.

\emph{Case 1:} When $\ell \le \frac {m-3} 2$, we have $m+2\ell \le 2m-3$ and $$i+j2^{2\ell+1} <n,$$
so \eqref{IJE} is impossible.

\emph{Case 2:} When $\ell=\frac {m-1} 2$, we have $i+j2^{2\ell+1} \equiv \Delta \pmod n$, where
\begin{eqnarray*}\Delta&=&j_{m-2}2^{2m-2}+(j_{m-3}-j_{m-1})2^{2m-3}+j_{m-4}2^{2m-4}+(j_{m-5}-j_{m-1})2^{2m-5}+ \cdots+
j_12^{m+1}  \\ && +(j_0-j_{m-1})2^m+i_{m-1}2^{m-1}+(i_{m-2}-j_{m-1})2^{m-2}+i_{m-3}2^{m-3}+\cdots+i_2 2^2+(i_1-j_{m-1})2+i_0.\end{eqnarray*}
It is clear from \eqref{IJO} that $\Delta=0 \text{ or } n$.

If $\Delta=0$, then we have
$$i_0=i_2=\cdots=i_{m-1}=0, \ \ i_1-j_{m-1}=i_3-j_{m-1}=\cdots=i_{m-2}-j_{m-1}=0,$$
$$j_1=j_3=\cdots=j_{m-2}=0, \ \ j_0-j_{m-1}=j_2-j_{m-1}=\cdots=j_{m-3}-j_{m-1}=0.$$
Thus
$$i_1=i_3=\cdots=i_{m-2}=1 \text{ and } j_0=j_2=\cdots=j_{m-1}=1.$$
It follows that
$$i=\frac {2^m-2} 3 \text{ and } j=\frac {2^{m+1}-1} 3.$$

If $\Delta=n$, then we have
$$i_0=i_2=\cdots=i_{m-1}=1, \ \ i_1-j_{m-1}=i_3-j_{m-1}=\cdots=i_{m-2}-j_{m-1}=0,$$
$$j_1=j_3=\cdots=j_{m-2}=1, \ \  j_0-j_{m-1}=j_2-j_{m-1}=\cdots=j_{m-3}-j_{m-1}=0.$$
\begin{itemize}
  \item For $i_1=1$, one can see that
  $$i_3=i_5=\cdots=i_{m-2}=1 \text{ and } j_0=i_2=\cdots=j_{m-3}=j_{m-1}=1.$$
  This leads to
  $$i=j=2^m-1.$$
  \item For $i_1=0$, it is easy to obtain that
  $$i_3=i_5=\cdots=i_{m-2}=0 \text{ and } j_0=i_2=\cdots=j_{m-3}=j_{m-1}=0.$$
  It then follows that
  $$i=\frac {2^{m+1}-1} 3 \text{ and } j=\frac {2^m-2} 3.$$
\end{itemize}

\emph{Case 3:} When $\ell \le \frac {m+1} 2$, let $\ell=\frac {m-1} 2+\epsilon$, where $1 \le \epsilon \le \frac {m-1} 2$.
In this case, one can show that $i+j2^{2\ell+1} \equiv \Delta \pmod n$, where
\begin{eqnarray*}\Delta&=&j_{m-2\epsilon-2}2^{2m-2}+\cdots+j_02^{m+2\epsilon}+i_{m-1}2^{m-1}+\cdots+i_{2\epsilon+1}2^{2\epsilon+1}+
(i_2\epsilon+j_{m-1})2^{2\epsilon}\\ &&+\cdots+(i_1+j_{m-2\epsilon})2+(i_0+j_{m-2\epsilon-1})<n.\end{eqnarray*}
Thus \eqref{IJO} does not hold.

By the discussions in Cases 1, 2, and 3, we have
$$J^+(\delta) \bigcap J^-(\delta)=\begin{cases}
\emptyset, & \text{ if } 2 \le \delta \le \frac {2^{m+1}-1} 3; \\
C_{\frac {2^m-2} 3} \cup C_{\frac {2^{m+1}-1} 3}, & \text{ if } \frac {2^{m+1}-1} 3+1 \le \delta \le 2^m-1; \\
C_{\frac {2^m-2} 3} \cup C_{\frac {2^{m+1}-1} 3} \cup C_{2^m-1}, & \text{ if } \delta=2^m, \end{cases}$$
where the unions are disjoint. The desired conclusion then follows from Lemma \ref{LMO}.
\end{proof}

\emph{B. Parameters of cyclic Hermitian LCD  codes $\mathcal C$}

Recall that $Q=4$ and $n=\frac {4^m-1} 3$. Let $\mathcal C$ be the cyclic Hermitian LCD  code of length $n$ with generator polynomial $g(x)$ given by $\eqref{G2}$.
Based on the preparations above on $4$-cyclotomic cosets modulo $n$, the dimension of the code $\mathcal C$ can be deduced.

\begin{theorem} \label{thm-QUA} Let $m \ge 2$ and $2 \le \delta \le 2^m$. Let $\mathcal C$ be the quaternary cyclic Hermitian LCD  code of length $n$ with generator polynomial $g(x)$ given by $\eqref{G2}$.
Then we have the following.
\begin{enumerate}
  \item When $m \ge 2$ is even, then $\mathcal C$ has length $n=\frac {4^m-1} 3$, dimension $k$, and minimum distance $d \ge \delta+1+\lfloor \frac {\delta-1} 2 \rfloor$, where
$$k=\begin{cases}
\frac {4^m-1} 3-2(\delta-\lfloor \frac {\delta-1} 4\rfloor-1)m-1, & \text{ if } 2 \le \delta \le \frac {2^{m+1} -2} 3; \\
\frac {4^m-1} 3-2(\delta-\lfloor \frac {\delta-1} 4\rfloor-2)m-1, & \text{ if } \delta=\frac {2^{m+1}+1} 3; \\
\frac {4^m-1} 3-2(\delta-\lfloor \frac {\delta-1} 4\rfloor-3)m-1, & \text{ if } \frac {2^{m+1}+1} 3+1 \le \delta \le \frac {2^{m+1}+2^{m-1}-1} 3; \\
\frac {4^m-1} 3-2(\delta-\lfloor \frac {\delta-1} 4\rfloor-4)m-1, & \text{ if } \frac {2^{m+1}+2^{m-1}-1} 3+1 \le \delta \le 2^m. \end{cases}$$
  \item  When $m \ge 5$ is odd, then $\mathcal C$ has length $n=\frac {4^m-1} 3$, dimension $k$, and minimum distance $d \ge \delta+1+\lfloor \frac {\delta-1} 2 \rfloor$, where
$$k=\begin{cases}
\frac {4^m-1} 3-2(\delta-\lfloor \frac {\delta-1} 4\rfloor-1)m-1, & \text{ if } 2 \le \delta \le \frac {2^{m+1} -1} 3; \\
\frac {4^m-1} 3-2(\delta-\lfloor \frac {\delta-1} 4\rfloor-2)m-1, & \text{ if } \delta=\frac {2^{m+1}+2} 3; \\
\frac {4^m-1} 3-2(\delta-\lfloor \frac {\delta-1} 4\rfloor-3)m-1, & \text{ if } \frac {2^{m+1}+2} 3+1 \le \delta \le \frac {2^{m+1}+2^{m-1}+1} 3; \\
\frac {4^m-1} 3-2(\delta-\lfloor \frac {\delta-1} 4\rfloor-4)m-1, & \text{ if } \frac {2^{m+1}+2^{m-1}+1} 3+1 \le \delta \le 2^m-1; \\
\frac {4^m-1} 3-2(\delta-\lfloor \frac {\delta-1} 4\rfloor-\frac 9 2)m-1, & \text{ if } \delta=2^m. \end{cases}$$
\end{enumerate}
\end{theorem}

\begin{proof}                                                                                                                                                                                                                                   One can  check that the following hold:
\begin{itemize}
  \item $|C_{n-2i}|=|C_i|$,
  \item $C_{n -2i}=C_{n-2j}$ if and only if $C_i=C_j$,
  \item $C_{n -2Qi}=C_{n -2i}$.
\end{itemize}

When $m$ is even, by Lemma \ref{LME}, we have
$$|J^+(\delta)|=|J^-(\delta)|=\begin{cases}
(\delta-1-\lfloor \frac {\delta-1} 4\rfloor)m, & \text{ if } 2 \le \delta \le \frac {2^{m+1} +1} 3; \\
(\delta-2-\lfloor \frac {\delta-1} 4\rfloor)m, & \text{ if } \frac {2^{m+1}+1} 3+1 \le \delta \le 2^m. \end{cases}$$
The degree of generator polynomial $g(x)$ of $\mathcal C$ is equal to
$$1+|J^+(\delta)|+|J^-(\delta)|-|J^+(\delta) \bigcap J^-(\delta)|.$$
The dimension $k$ of the code $\mathcal C$ then follows from Proposition \ref{PME}.
It is known that $C_{n-2i}=C_{n-i/2}$ if $2 \mid i$. Then the set
$$\{n-\lfloor \frac {\delta-1} 2 \rfloor, n -\lfloor \frac {\delta-1} 2 \rfloor+1, \ldots,  n-1, 0, n+1, n+2,\ldots, n+\delta-1\}$$ is
contained in the defining set of $\mathcal C$. The desired conclusion on the minimum distance is
then derived from Lemma \ref{LBCH}.

For odd $m$, the dimension $k$ of the code $\mathcal C$ follows from Lemma \ref{LMO} and Proposition \ref{PMO} and
the lower bound on the minimum distance can be similarly obtained.
\end{proof}

\section{Hermitian orthogonal direct sum masking}

In this section, we present an application of Hermitian LCD codes in Hermitian orthogonal direct sum masking (ODSM),
which can ensure a protection against fault injection attack (FIA). It should be
remarked that ODSM was introduced in \cite{BCCGM, CG}.

Let $\gf(Q)^n$ be a vector space over $\gf(Q)$ of dimension $n$ and
$$\gf(Q)^n=\mathcal C \oplus \mathcal D$$
for two  supplementary vector subspaces $\mathcal C$ and $\mathcal D$.
Suppose that the dimensions of $\mathcal C$ and $\mathcal D$ are $k$ and $n-k$.
For a sensitive data $x$ with $k$ information symbols, let $c=xG$, where $G$ is a
generator matrix of $\mathcal C$. To protect $x$, $n-k$ random information symbols are required, which are denoted by
$y$. Let $G'$ be a generator matrix of $\mathcal D$. Denote
$$z=c + d,$$
where $c=xG$ is called the coded sensitive data and $d=yG'$ is called the mask. In fact, $x \in \gf(Q)^k$ and $y \in \gf(Q)^{n-k}$.

The question is how to  recover $x$ from the state $z$. When $\mathcal C$ is a Hermitian LCD code, i.e., $\mathcal D=\mathcal C^{\perp_\h}$,
$G'=H$, where $H$ is the check matrix of $\mathcal C$. In this case, we then have $GH^{\perp_\h}=0$ and
\begin{equation} \label{zxy-DS} z=xG + yH.\end{equation}
Hence, the sensitive $x$ and random $y$ are recovered from $z$ as follows:
\begin{equation} \label{zx-DS} x=zG^{\perp_\h} (GG^{\perp_\h})^{-1}, \end{equation}
\begin{equation} \label{zy-DS} y=zH^{\perp_\h} (HH^{\perp_\h})^{-1}. \end{equation}

Given a FIA, the state $z$ may be modified into $z+\epsilon$. By \eqref{zxy-DS}, we similarly have
$\epsilon=eG+fH$, where $e \in \gf(Q)^k$ and $f \in \gf(Q)^{n-k}$. The task now is to check whether or not the modification appears.
For the purpose of the protection of the sensitive data $x$, we should avoid computing $x$ by \eqref{zx-DS}. The variable $y$ is random and
 does not carry any exploitable information. Hence, the harmless verification strategy is to check the following:
 $$(z+\epsilon)H^{\perp_\h} (HH^{\perp_\h})^{-1}\stackrel{?}{=} y.$$
It is clear that the equality holds if and only if $f=0$, i.e., $\epsilon \in \mathcal C$.
We conclude that the fault can not be detected if $\epsilon \in \mathcal C \setminus \{0\}$. In fact, we have proved the following theorem.
\begin{theorem} \label{thm-FIA} Let $d$ be the minimum distance of $\mathcal C$.
The fault $\epsilon$ must be detected if its Hamming weight  is less than $d$.
\end{theorem}

Theorem \ref{thm-FIA} indicates that the Hermitian LCD codes with
 great minimum distances are desirable to improve the resistance against FIA.

\begin{example}
Let $\mathcal C$ be the quaternary Hermitian LCD code of parameters $[9, 2, 6]$, which was given in the first case of Example \ref{exa-1}. The generator polynomial of $\mathcal C$ is
$$g(x)=X^7 + X^6 + X^4 + X^3 + X + 1 \in \gf(4)(X)$$ and the corresponding generator matrix is
$$G=\left(\begin{array}{ccccccccc}
      1 & 1 & 0 & 1 & 1 & 0 & 1 & 1 & 0 \\
      0 & 1 & 1 & 0 & 1 & 1 & 0 & 1 & 1
      \end{array}\right)_{2 \times 9}.$$
It then follows that the check matrix of $\mathcal C$ is
$$H=\left(\begin{array}{ccccccccc}
      1 & 1 & 1 & 0 & 0 & 0 & 0 & 0 & 0 \\
      0 & 1 & 1 & 1 & 0 & 0 & 0 & 0 & 0 \\
      0 & 0 & 1 & 1 & 1 & 0 & 0 & 0 & 0 \\
      0 & 0 & 0 & 1 & 1 & 1 & 0 & 0 & 0 \\
      0 & 0 & 0 & 0 & 1 & 1 & 1 & 0 & 0 \\
      0 & 0 & 0 & 0 & 0 & 1 & 1 & 1 & 0 \\
      0 & 0 & 0 & 0 & 0 & 0 & 1 & 1 & 1
      \end{array}\right)_{7 \times 9}.$$

      Let $\gf(4)=\{0, 1, \omega, \omega^2\}$, where $\omega$ is a root of $X^2+X+1 \in \gf(2)(X)$.
      Let $x=(1,\omega)$ and $y=(1,1,1,1,1,1)$. Then
      $$z=xG+yH=(0, \omega^2, \omega^2, 0, \omega, \omega^2, 0, \omega^2, \omega^2).$$
      \begin{enumerate}
        \item If $\epsilon_1=(\omega,0,\omega,\omega,0,\omega,\omega,0,\omega)$, then
      $$(z+\epsilon_1)H^{\perp_\h} (HH^{\perp_\h})^{-1}=(1,1,1,1,1,1,1)=y.$$ In this case, $\epsilon_1$ can not be detected as $\epsilon_1 \in \mathcal C$.
        \item If $\epsilon_2=(0,1,0,0,1,0,0,0,0)$, then
      $$(z+\epsilon_2)H^{\perp_\h} (HH^{\perp_\h})^{-1}=(1,0,0,1,1,1,1)\ne y.$$ It then follows that
      $\epsilon_2$ is detected  as the Hamming weight of $\epsilon_2$ is less than $6$ (minimum distance of $\mathcal C$).
      \end{enumerate}
\end{example}

\section{Concluding remarks}

The main contributions of this paper are the following:
\begin{enumerate}
    \item A general construction of cyclic Hermitian LCD  codes via generator polynomials was given in Theorem \ref{TH}.
    \item The constructions of Hermite LCD cyclic codes of three different lengths were presented and their parameters were investigated (See Theorems \ref{thm-HOP}, \ref{T1}, \ref{thm-QUA}).
    The dimensions of Hermite LCD cyclic codes were settled and the lower bounds on their minimum distances were obtained from the BCH bound.
    \item An application of Hermitian LCD codes in Hermitian ODSM was proposed,
which can ensure a protection against FIA.
\end{enumerate}

Generally, it is an extremely hard work to determine the minimum distances of these cyclic Hermitian LCD codes.
The importance of the minimum distance of the Hermitian LCD codes was exhibited in providing the protection against FIA.
The reader is thus cordially invited to
study the minimum distances of the codes presented in this paper.

\section*{Acknowledgements}
The author is very grateful to
Prof. Cunsheng Ding for his kindly help to complete this work.

\end{document}